\documentclass{article}

\usepackage[margin=1.5in]{geometry}

\usepackage[utf8]{inputenc} 
\usepackage[T1]{fontenc}    
\usepackage{hyperref}       
\usepackage{url}            
\usepackage{booktabs}       
\usepackage{verbatim}
\usepackage{natbib}
\bibpunct[; ]{(}{)}{,}{a}{}{;}

\usepackage{amsmath}
\usepackage{amssymb}
\usepackage{amsthm}
\usepackage{amsfonts}       
\usepackage{nicefrac}       
\usepackage{dsfont}
\usepackage{placeins}
\usepackage{authblk}

\usepackage{tikz,pgfplots,pgfplotstable,xcolor}
\pgfplotsset{width=7cm,compat=1.8}  

\usepackage{algorithm}
\makeatletter
\renewcommand{\ALG@name}{Algorithm}
\makeatother
\usepackage{algpseudocode}
\algblockdefx{ForAllP}{EndFaP}[1]%
  {\textbf{for all }#1 \textbf{do in parallel}}%
  {\textbf{end for}}

\newcommand{\Scal}{\mathcal{S}}

\newcommand{\Ycal}{\mathcal{Y}}
\newcommand{\Zcal}{\mathcal{Z}}

\newcommand{\supp}{\operatorname{supp}}

\DeclareMathOperator{\E}{\mathbb{E}}
\DeclareMathOperator*{\argmin}{arg\,min}
\DeclareMathOperator*{\argmax}{arg\,max}
\DeclareMathOperator*{\bigtimes}{\textnormal{\Large $\times$}}

\newtheorem{theorem}{Theorem}
\newtheorem{proposition}{Proposition}

\newtheorem{conjecture}{Conjecture}

\theoremstyle{definition}
 
\newtheorem{example}{Example}

\usepackage{graphicx}
\graphicspath{{../Figures/} {../Computations/}}    

\usepackage{siunitx,caption}
\usepackage{subfig}
\usepackage{threeparttable}

\definecolor{ciColor}{HTML}{89E894}
\definecolor{siColor}{HTML}{BADD30} 
\definecolor{uiYColor}{HTML}{34DDDD} 
\definecolor{uiZColor}{HTML}{C4FCFC} 

\pgfplotsset{every axis legend/.style={
		cells={anchor=center},
		inner xsep=3pt,inner ysep=2pt,nodes={inner sep=2pt,text depth=0.15em},
		anchor=north west,
		shape=rectangle,
		fill=white,
		draw=black,
		font=\footnotesize,
		at={(rel axis cs:0.02,0.98)}
	}
}
\pgfplotsset{scaled y ticks=false}  

\pgfplotstableread{ 
	Label                                SI CI UIy UIz
	hours-per-week\space($Y$),\space occupation\space($Z$)   0.3537  0.1908  0.0 0.4556
	age\space($Y$),\space sex\space($Z$)   0.3207  0.2716  0.4078 0.0
	education\space($Y$),\space occupation\space($Z$)   0.3892  0.1472  0.0 0.4637
	race\space($Y$),\space occupation\space($Z$)   0.0851  0.0588  0.0 0.8561
	education\space($Y$),\space race\space($Z$)   0.1657  0.1321  0.7022 0.0
	education\space($Y$),\space sex\space($Z$)  0.4054  0.4583  0.1081 0.0282
}\testdata

\usepackage{blkarray}
\newcommand{\matindex}[1]{\mbox{\scriptsize#1}}
\usepackage{mathtools}

\title{\bf Computing the Unique Information} 

\author[1]{Pradeep Kr.~Banerjee}
\author[1]{Johannes Rauh}
\author[1,2]{Guido Mont\'ufar}
\affil[1]{Max Planck Institute for Mathematics in the Sciences, Leipzig, Germany}
\affil[2]{Departments of Mathematics and Statistics, UCLA, USA}

\date{}

\begin{document}

\maketitle

\begin{abstract} 
Given a pair of predictor variables and a response variable, how much information do the predictors have about the response, and how is this information distributed between unique, redundant, and synergistic components? Recent work has proposed to quantify the unique component of the decomposition as the minimum value of the conditional mutual information over a constrained set of information channels. We present an efficient iterative divergence minimization algorithm to solve this optimization problem with convergence guarantees and evaluate its performance against other techniques.  
\smallskip


\end{abstract}

\section{Introduction}

When Shannon proposed to use entropy in order to quantify information, he had in mind a very specific setting of communication over a noisy channel. Since then, the use of entropic quantities has been greatly expanded, with successful applications in statistical physics, complex systems, neural networks, and machine learning. In particular, transfer entropy is used as a tool to study causality in dynamical systems~\citep{TE_Schreiber}, and mutual information as a criterion in feature selection~\citep{VergaraEstevez14:Review_MI_feature_selection}. In many applications, the effort of estimating entropic quantities, which may be considerable, is out-weighed by the performance gain.

Despite the success of information theory, there are still many open questions about the nature of information. In particular, since information is not a conservation quantity, it is difficult to describe how information is distributed over composite systems. Clearly, different subsystems may have exclusive (or unique) information, or they may have redundant information. Moreover, synergy effects complicate the analysis: It may happen that some information is not known to any subsystem but can only be recovered from knowledge of the entire system. An example is the checksum of several digits, which can only be computed when all digits are known. Such synergy effects abound in cryptography, where the goal is that the encrypted message alone contains no information about the original message without knowledge of the key. It is also believed that synergy plays a major role in the neural code~\citep{LathamNirenberg05:Synergy_and_redundancy_revisited}. 

In spite of their conceptual importance, so far there is no consensus on how to measure or extract the unique, shared, and synergistic portions of joint information, even though there have been several proposals~\citep[e.g.,][]{McGill1954,Bell2003}. \citet{WilliamsBeer} proposed a principled approach to decomposing the total mutual information of a system into positive components corresponding to a lattice of subsystems. This was followed up by the axiomatic approach from~\cite{e16042161}, quantifying the unique, shared, and synergistic information based on ideas from decision theory. In the past couple of years, the latter approach has steadily gained currency with applications ranging, for instance, from quantifying the neural code~\citep{pica2017quantifying}, to learning deep representations~\citep{minsyn2017}. We focus on their definitions explained in Section~\ref{sec:quantifying}. 

Although theoretically promising, these definitions involve an optimization problem that complicates experimentation and applications. Indeed, \cite{e16042161} note that the optimization problem, although convex, can be very ill conditioned, 
and difficulties have been reported with out-of-the-box methods or custom implementations either failing to produce the correct results, or taking extremely long to converge. \cite{DOT2017bivariate} compares different approaches to solving the problem using a number of off-the-shelf software packages. They note that subgradients of the objective function do not exist everywhere on the boundary of the feasible set, and hence, algorithms such as projected (sub-)gradient descent and generic active set methods are not particularly good choices for solving the problem.

In Section~\ref{sec:computing}, we derive an alternating divergence minimization algorithm for solving the optimization problem in~\cite{e16042161} with convergence guarantees. It is similar in spirit to the Blahut-Arimoto algorithm~(BAA)~\citep{1054855,arimoto1972} which is also based on alternating optimization. However, there are significant differences, especially in relation to the nature of the constraints. In particular, there is no direct relation between the two algorithms, and we do not see a way to phrase the computation as a capacity or rate distortion computation. In Section~\ref{sec:experiments}, we present computational results comparing our algorithm with other approaches. Matlab and Python implementations of our algorithm are available online\footnote{\url{https://github.com/infodeco/computeUI}}. Our algorithm consistently returns accurate solutions, while still requiring less computation time than other methods. 
We wrap up and give a brief outlook in Section~\ref{sec:discussion}. 
Relevant notations are included in Appendix~\ref{sec:notation}. 

\section{Quantifying the unique information}
\label{sec:quantifying}
While much research has focused on finding an information measure for a single aspect (like synergy), the seminal paper by~\citet{WilliamsBeer} introduced an approach to find a complete decomposition of the total mutual information $I(S ; Y_1,\ldots,Y_k)$ about a signal $S$ that is distributed among a family of random variables $Y_1 ,\ldots, Y_k$. 
Here, the total mutual information is expressed as a sum of non-negative terms with a well-defined interpretation corresponding to the different ways in which information can have aspects of redundant, unique, or synergistic information. For example, in the case $k = 2$, writing $Y_1\equiv Y$ and $Y_2\equiv Z$, the decomposition is of the form
\begin{equation}
I(S; Y,Z) = 
\underbrace{SI(S; Y,Z)}_{\text{shared (redundant)}} 
+ \underbrace{CI(S; Y,Z)}_{\text{complementary (synergistic)}} 
+ \underbrace{UI(S;Y \backslash Z)}_{\text{unique $Y$ wrt $Z$}} 
+ \underbrace{UI(S;Z \backslash Y)}_{\text{unique $Z$ wrt $Y$}}, 
\label{eq:MI-decomposition} 
\end{equation}
where $SI(S;Y,Z)$, $CI(S;Y,Z)$, $UI(S;Y \backslash Z)$, and $UI(S;Z \backslash Y)$ are nonnegative functions that depend continuously on the joint distribution of $(S,Y,Z)$. 
Furthermore, these functions are required to satisfy 
the intuitive equations
\begin{equation}
\label{eq:MI-decomposition-2}
\begin{aligned}
I(S; Y) &= SI(S; Y,Z) + UI(S; Y \backslash Z),\\
I(S; Z) &= SI(S; Y,Z) + UI(S;Z \backslash Y). 
\end{aligned} 
\end{equation}

Combining these equations, it follows that the co-information can be written as the difference of redundant and synergistic information, which agrees with the general interpretation of co-information: 
\begin{equation}
CoI(S; Y; Z) 
:= I(S; Y) - I(S; Y|Z) = SI(S; Y,Z) - CI(S; Y,Z).
\end{equation}
Similarly, the conditional mutual information satisfies
\begin{equation}
I(S; Y|Z) = I(S; Y,Z) - I(S; Z) = CI(S;Y,Z) + UI(S;Y \backslash Z).
\end{equation}
The decomposition is illustrated in Figure~\ref{fig:census}a. 

Although the above framework is very appealing, there is no general agreement on how to define the corresponding functions for shared, unique, and synergistic information. When~\cite{WilliamsBeer} presented their information decomposition framework, they also proposed specific measures. However, their functions have been criticized as overestimating redundant and synergistic information, while underestimating unique information~\citep{GriffithKoch2014:Quantifying_Synergistic_MI}%
\footnote{For example, in the case of two independent variables $Y,Z$ and a copy $S=(Y,Z)$, the measure of Williams and Beer assigns $0$ bit of unique information to $Y$ and~$Z$, and all available information is interpreted as either redundant or synergistic.}. Another proposal of information measures for the bivariate case ($k=2$) that involves information-geometric ideas is presented in~\cite{HarderSalgePolani2013:Bivariate_redundancy}. 

Here we follow the approach from~\citet{e16042161} and use the functions $SI$, $UI$ and $CI$ defined there, since it is the most principled approach, based on ideas from decision theory and having an axiomatic characterization. This approach covers only $k=2$, but situations with larger~$k$ can be analyzed by grouping the variables. The decomposition is based on the idea that unique and shared information, $UI(S;Y \backslash Z)$ and $SI(S; Y,Z)$, should depend only on the marginal distributions of the pairs $(S,Y)$ and $(S,Z)$.  
It gives similar values as the functions defined in~\cite{HarderSalgePolani2013:Bivariate_redundancy}. Incidentally, the bivariate synergy measure derived from this approach agrees with the synergy measure defined by~\cite{GriffithKoch2014:Quantifying_Synergistic_MI} for arbitrary~$k$.  

For some finite state spaces $\Ycal,\Zcal,\Scal$, let $\mathbb{P}_{\Scal\times\Ycal\times\Zcal}$ be the set of all joint distributions of $(S,Y,Z)$. 
Given $P \in \mathbb{P}_{\Scal\times\Ycal\times\Zcal}$, let 
\begin{align}
\Delta_P := \big\{Q \in \mathbb{P}_{\Scal\times\Ycal\times\Zcal}\colon  Q_{SY}(s,y)=P_{SY}(s,y) \text{ and } Q_{SZ}(s,z)=P_{SZ}(s,z)\big\}
\label{eq:delP}
\end{align}
denote the set of joint distributions of $(Y,Z,S)$, that have the same marginals on $(S,Y)$ and $(S,Z)$ as $P$. 
\cite{e16042161} define the unique information that $Y$ conveys about $S$ with respect to $Z$ as 
\begin{align}
UI(S;Y\backslash Z) := \min_{Q \in \Delta_P} I_Q(S;Y|Z). 
\label{eq:uidefinition}
\end{align}
See Appendix~\ref{sec:properties} for a brief discussion of the properties of the unique information explaining why this definition makes sense. 
By~\eqref{eq:MI-decomposition} and~\eqref{eq:MI-decomposition-2}, 
specifying~\eqref{eq:uidefinition} fixes the other three functions in~\eqref{eq:MI-decomposition}, which are then 
\label{subeq:auxdefinitions}
\begin{align}
  UI(S;Z \backslash Y) &:= \min_{Q \in\Delta_P} I_Q(S;Z|Y), \label{subeq:uizdef} \\ 
  SI(S;Y,Z) &:= I(S;Y) - \min_{Q \in \Delta_P} I_Q(S;Y|Z) = \max_{Q\in\Delta_{P}} CoI_{Q}(S;Y;Z), \label{subeq:sidef} \\
  CI(S;Y,Z) &:= SI(S;Y,Z) - CoI(S;Y;Z) = I(S;Y,Z) - \min_{Q\in\Delta_{P}}I_{Q}(S;Y,Z). \label{subeq:cidef}
\end{align}
Since $\Delta_{P}$ is compact and the mutual information is a continuous function, these maxima and minima are all well-defined. 

We briefly illustrate the mutual information decomposition~\eqref{eq:MI-decomposition} by evaluating it on a real data set.
\begin{example}[\emph{The US 1994 census income data set~\citep{Lichman:2013}}]
The task is to relate a list of predictor variables with a binary response variable. The predictors include: sex (binary: Male, Female), age (continuous variable divided into 4 categories: $<24$, $\numrange{24}{35}$, $\numrange{36}{50}$, $>50$), race (5 values: White, Asian-Pac-Islander, Amer-Indian-Eskimo, Black, Other), education level (4 values: Basic-schooling, Attended-HS, Bachelors-and-above, Vocational), occupation (14 values: Tech-support, Craft-repair, Other-service, etc.), and hours-per-week (continuous variable grouped into 2 categories: $\le40$, $>40$). The response is the yearly income, with values $>50$K and $\leq50$K. 

Figure~\ref{fig:census} shows the evaluation of the information decomposition~\eqref{eq:uidefinition}--\eqref{subeq:cidef} on this data set, computed using the algorithm that we will present in Section~\ref{sec:computing}. We see, for instance, that most of the information that race and occupation convey about income, is uniquely in the occupation. On the other hand, education and sex have about equally large shared and complementary components. Age conveys a large unique information about income with respect to sex, as does occupation with respect to hours-per-week. These results appear quite reasonable. They illustrate how the decomposition allows us to obtain a fine-grained quantitative analysis of the relationships between predictors and responses. 
\end{example}

\begin{figure}[tb]
\centering
\subfloat[]{
	\resizebox {.3\textwidth} {!} {
\begin{tikzpicture} 
\useasboundingbox (-2.2,-2.2) rectangle (2.2,2.2);

\begin{scope}
\clip (0,0) circle (2cm);
\draw [fill=ciColor, fill opacity=1] (0,0) circle (2cm);
\draw [fill=uiYColor, fill opacity=1] (-1.25,-1.5) circle (2cm);
\draw [fill=uiZColor, fill opacity=1] (1.25,-1.5) circle (2cm);
\end{scope}

\begin{scope}
\clip (-1.25,-1.5) circle (2cm);
\clip (1.25,-1.5) circle (2cm);
\draw [fill=siColor, fill opacity=1] (0,0) circle (2cm);
\end{scope}

\begin{scope}
\clip (0,0) circle (2cm);
\draw [fill=none, fill opacity=1] (-1.25,-1.5) circle (2cm);
\draw [fill=none, fill opacity=1] (1.25,-1.5) circle (2cm);
\end{scope}
\draw [fill=none, fill opacity=1] (0,0) circle (2cm);

\draw (0,1) node {$CI$};
\draw (-1.25,-.3) node {$UI_{Y\setminus Z}$};
\draw (1.25,-.3) node {$UI_{Z\setminus Y}$};
\draw (0,-1) node {$SI$};

\draw [fill=none, color=white] (0,0) circle (2.2cm);
\end{tikzpicture}}}

\subfloat[]{	
	\begin{tikzpicture}
	\begin{axis}[
	xbar stacked, 
	xmin=0,         
	xmax=1,
	legend style={
		legend columns=4,
		at={(xticklabel cs:0.5)},
		anchor=north,
		draw=none,
		axis line style={draw=none},
		tick style={draw=none}
	},
	legend style={font=\footnotesize},
	area legend,
	xtick=\empty,
	ytick=data,     
	yticklabels from table={\testdata}{Label}  
	]
	\addplot [fill=siColor] table [x=SI, meta=Label,y expr=\coordindex] {\testdata};  
	\addplot [fill=ciColor] table [x=CI, meta=Label,y expr=\coordindex] {\testdata};
	\addplot [fill=uiYColor] table [x=UIy, meta=Label,y expr=\coordindex] {\testdata};
	\addplot [fill=uiZColor] table [x=UIz, meta=Label,y expr=\coordindex] {\testdata};
	\legend{$SI$, $CI$, $UI(S;Y\backslash Z)$, $UI(S;Z\backslash Y)$}
	\end{axis}
	\end{tikzpicture}}
\caption{(a)~Illustration of the decomposition~\eqref{eq:MI-decomposition} of the mutual information of a pair $(Y,Z)$ and $S$ into the complementary (synergistic) information $CI$, the unique information $UI$ of $Y$ with respect to $Z$ and conversely, and the shared (redundant) information $SI$. (b)~Information decomposition evaluated on the US 1994 census data set~\citep{Lichman:2013}: the attributes $Y$ and $Z$ predict the income category $S$ ($>$ or $\le$ \$50K per year). Each bar is normalized by the total mutual information $I(S;Y,Z)$ to highlight the relative values of $SI$, $CI$ and $UI$.}
\label{fig:census}
\end{figure}
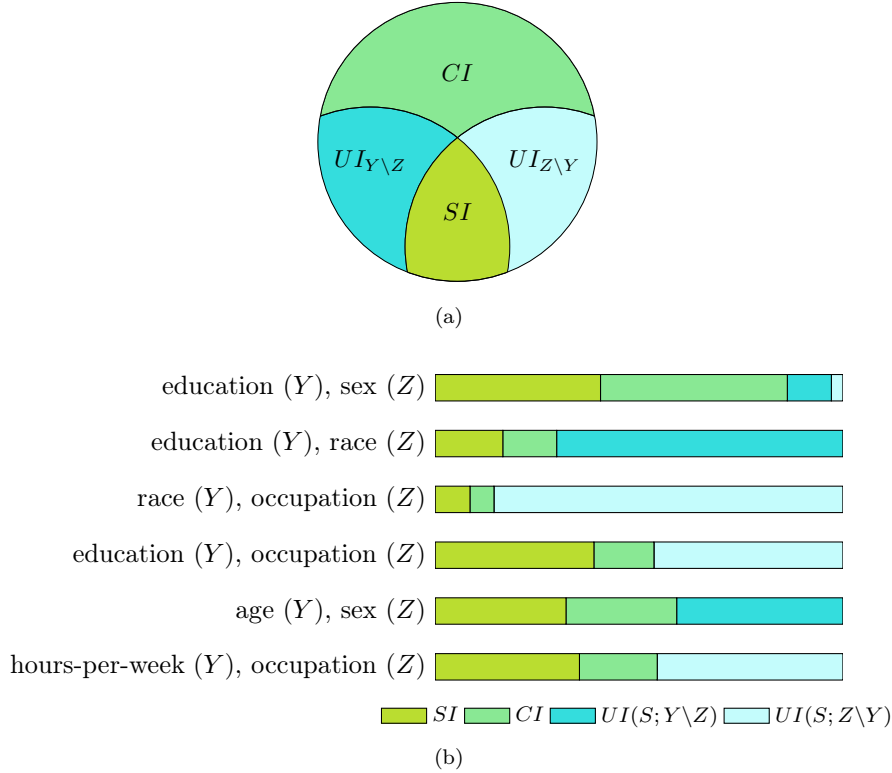	

\section{Computing the information decomposition}
\label{sec:computing}
We need to solve only one of the optimization problems~\eqref{eq:uidefinition}--\eqref{subeq:cidef} in order to obtain all the terms in the information decomposition. 
We first note that solving~\eqref{eq:uidefinition}--\eqref{subeq:cidef} is equivalent to solving an equivalent convex minimization problem, namely, that for a function which we call the \emph{union information}, defined as follows. 
\begin{align}
I_\cup(S;Y,Z) &:= I(S;Y,Z) - CI(S;Y,Z) = \min_{Q \in\Delta_P} I_Q(S;Y,Z).\label{eq:unioninfo}
\end{align}

We first note that this is a convex minimization problem: 
\begin{proposition}
	\label{prop:eqvoptm}
	The optimization problems~\eqref{eq:uidefinition}--\eqref{subeq:cidef} and~\eqref{eq:unioninfo} are convex. Also, they are equivalent in that a distribution $Q$ solves one of them if and only if it solves all of them.  
\end{proposition}
\begin{proof}
	The equivalence of the optimization problems follows from~\eqref{eq:MI-decomposition}--\eqref{eq:MI-decomposition-2}. Moreover, $$\min_{Q\in\Delta_{P}}I_{Q}(S;Y,Z) = H(S) - \max_{Q\in\Delta_{P}} H_{Q}(S|Y,Z),$$ since $H(S)$ is constant on~$\Delta_{P}$. Convexity of the optimization problems follows from the fact that $H_{Q}(S|Y,Z)$ is concave with respect to~$Q$~\citep{ckbook}.
\end{proof}
The target function $I_Q(S;Y,Z)$ is convex, but not strictly convex; it is continuous in $\Delta_P$, and smooth in the interior of $\Delta_P$ (but not on the boundary). For certain~$P$, the optimal~$Q\in\Delta_{P}$ may not be unique. However, convexity guarantees that any local optimizer is a global optimizer and that the optimum value is unique.

\paragraph{Double minimization formulation.}
The mutual information can be written in the form
$I_P(S;Y,Z) =\mathop{\min}_{R_{YZ} \in \mathbb{P}_{\Ycal\times \Zcal}} D(P \|P_S R_{YZ})$, 
with the minimum attained at 
$R^*_{YZ}=P_{YZ}$~\cite[see, e.g.,][eq.~(8.7)]{ckbook}. 

With this expression, we can rewrite~\eqref{eq:unioninfo} as a double minimization problem: 
\begin{align}
I_\cup(S;Y,Z) &
= \min_{Q \in\Delta_P} \min_{R_{YZ}\in\mathbb{P}_{\Ycal\times\Zcal}} D(Q \|Q_S R_{YZ}). 
\label{eq:variationalrep1}
\end{align}

\paragraph{Conditional probability formulation.}
The minimization problem~\eqref{eq:variationalrep1} can also be studied and solved over a set of conditional probabilities, instead of the set $\Delta_{P}$ that consists of joint probability distributions. 
In fact, $\Delta_{P}$ is in bijection with
$\Delta_{P,S} := \bigtimes_{s\in\Scal}\Delta_{P,s}$,  
where 
\begin{align}
\Delta_{P,s} := \big\{Q_{YZ}\in \mathbb{P}_{\Ycal\times\Zcal}\colon Q_{Y}(y)=P_{Y|S}(y|s) \text{ and } Q_{Z}(z)=P_{Z|S}(z|s)\big\}, \quad s\in\Scal. 
\label{eq:delPs}
\end{align}
The set $\Delta_{P,s}$ is the linear family of probability distributions of $(Y,Z)$ defined by fixing the marginal distributions of $Y$ and $Z$ to be those of $P_{YZ|s}$. 
Any joint distribution $Q\in\Delta_P$ has the form $Q=P_SQ_{YZ|S}$ with $Q_{YZ|S}\in\Delta_{P,S}$. 
In turn, the optimization problem~\eqref{eq:variationalrep1} can be written as 
\begin{align}
I_\cup(S;Y,Z) 
=&\min_{Q_{YZ|S} \in\Delta_{P,S}} \min_{R_{YZ}\in\mathbb{P}_{\Ycal\times\Zcal}} D(P_{S}Q_{YZ|S}\|P_{S}R_{YZ}) \notag\\
=&\min_{R_{YZ}\in\mathbb{P}_{\Ycal\times\Zcal}} \sum_s P_S(s) \min_{Q_{YZ|s} \in\Delta_{P,s}} D(Q_{YZ|s}\|R_{YZ}). 
\label{eq:variationalrep2} 
\end{align}

\paragraph{Alternating divergence minimization.}
With the formulation obtained above, we are able to derive an alternating optimization algorithm and leverage classic results to prove convergence and optimality. An alternating algorithm iteratively fixes one of the two free variables and optimizes over the other. Starting with some $R^{(0)}_{YZ}\in \mathbb{P}_{\Ycal\times\Zcal}$, recursively define
\begin{subequations}
\label{subeq:alternating}
\begin{align}
Q^{(i+1)}_{YZ|s}&=\argmin_{Q_{YZ|s}\in\Delta_{P,s}} D(Q_{YZ|s}\|R^{(i)}_{YZ}) \quad\text{for each $s\in\Scal$},\label{subeq:alternating1}\\
 R^{(i+1)}_{YZ}&=\argmin_{R_{YZ}\in\mathbb{P}_{\Ycal\times\Zcal}} D(P_{S}Q^{(i+1)}_{YZ|S}\|P_{S}R_{YZ}). \label{subeq:alternating2}
\end{align}
\end{subequations}

With suitable initialization, this iteration converges to a pair attaining the global optimum:

\begin{theorem}\label{Thc}
  Given $P\in\mathbb{P}_{\Scal\times\Ycal\times\Zcal}$ and an initial value $R^{(0)}_{YZ}\in
  \mathbb{P}_{\Ycal\times\Zcal}$ of full support, 
  the iteration~\eqref{subeq:alternating} converges. 
  Moreover, the limit $\lim_{i\to\infty}P_{S}Q^{(i)}_{YZ|S}$ is a global optimum of the minimization
  problem~\eqref{eq:variationalrep2}. 
\end{theorem}
\begin{proof}
For any $P$, the subsets $\Delta_P$ and $\{ P_SR_{YZ}\colon R_{YZ}\in \mathbb{P}_{\Ycal\times\Zcal} \}$ of $\mathbb{P}_{\Scal\times\Ycal\times\Zcal}$ are compact and convex. 
The statement then follows from~\cite[][Corollary 5.1]{csiszar2004}. 
\end{proof}

\paragraph{Implementation.}
Pseudocode for the alternating divergence minimization algorithm for computing the union information (\texttt{admUI}) is in Algorithm~\ref{alg:admUI}. 
We next discuss the two steps separately. 
\paragraph{\underline{Step 1}}
The optimization problems~\eqref{subeq:alternating1} are relatively standard and can be solved, e.g., using generalized iterative scaling (GIS) (pseudocode in Algorithm~\ref{alg:Iproj}). 
\begin{theorem}\label{theorem:innerloop}
The nonnegative functions $b_n$ on $\Ycal\times\Zcal$ defined recursively by 
\begin{equation}
b_0(y,z) = R_{YZ}(y,z), \quad b_{n+1}(y,z) = b_n(y,z) \bigg[\frac{P_{Y|S}(y|s)}{\sum_z b_n(y,z)}\bigg]^{1/2} \bigg[\frac{P_{Z|S}(z|s)}{\sum_y b_n(y,z)}\bigg]^{1/2}, 
\label{eq:Iprojiteration}
\end{equation}
converge to $\argmin_{Q_{YZ|s}\in\Delta_{P,s}}D(Q_{YZ|s}\|R_{YZ})$, that is, the $I$-projection of $R_{YZ}$ to $\Delta_{P,s}$. 
\end{theorem}
\begin{proof}
The claim follows from~\cite[][Theorem~5.2]{csiszar2004}. 
\end{proof}
\begin{algorithm}[H]
	\begin{algorithmic}[1]
		\State{\textbf{Input}: Marginals $P_{SY}$ and $P_{SZ}$}
		\State{\textbf{Output}: $Q^\ast = \argmin_{Q\in\Delta_P} I_Q(S;Y,Z)$}
		\State{\textbf{Initialization}: Some $R^{(0)}_{YZ}$ from the interior of $\mathbb{P}_{\Ycal\times\Zcal}$. Set $i=0$.}
		\While{not converged}
		\ForAllP{$s \in \supp(P_S)$}  
		\State{$Q^{(i+1)}_{YZ|s}\gets \argmin_{Q_{YZ|s}\in\Delta_{P,s}}D(Q_{YZ|s}\|R^{(i)}_{YZ|s})$\quad (Algorithm 2)} \Comment{Step 1}
		\EndFaP
		\State{$R^{(i+1)}_{YZ}(y,z)\gets \sum_{s\in \Scal}P_S(s)Q^{(i+1)}_{YZ|S}(y,z|s)$} \Comment{Step 2}
		\State{$i\gets i+1$}
		\EndWhile
		\State \textbf{return} $P_S Q^{(i)}_{YZ|S}$
	\end{algorithmic}
	\caption{\textbf{ 
			Alternating divergence minimization for the union information (\texttt{admUI})}}
	\label{alg:admUI}
\end{algorithm}
\vspace{-6mm}
\begin{algorithm}[H]
	\begin{algorithmic}[1]
		\State{\textbf{Input}: Marginals $P_{SY}$ and $P_{SZ}$, some $s\in \Scal$, and target distribution $R_{YZ}$}
		\State{\textbf{Output}: $Q^\ast_{YZ|s} = \argmin_{Q_{YZ|s}\in\Delta_{P,s}}D(Q_{YZ|s}\|R_{YZ})$}
		\State{\textbf{Initialization}: 
			$b_0(y,z) \gets R_{YZ}(y,z)$. Set $n=0$.}
		\While{not converged}
		\State $b_{n+1}(y,z) \gets b_n(y,z)\left[\frac{P_{Y|S}(y|s)}{\sum_z b_n(y,z)}\right]^{1/2}\left[\frac{P_{Z|S}(z|s)}{\sum_y b_n(y,z)}\right]^{1/2}$
		\State{$n\gets n+1$}
		\EndWhile\label{euclidendwhile}
		\State \textbf{return} $b_{n}$
	\end{algorithmic}
	\caption{\textbf{
			The $I$-projection of $R_{YZ}$ to $\Delta_{P,s}$}}
	\label{alg:Iproj}
\end{algorithm}

\paragraph{\underline{Step 2}}
Using the variational representation for $I_{P_SQ_{YZ|S}}(S;Y,Z)$ discussed prior~\eqref{eq:variationalrep1}, we can write the minimizer of~\eqref{subeq:alternating2} in closed form as 
\begin{equation}
 R^{(i+1)}_{YZ}(y,z) = \sum_{s\in \Scal}P_S(s)Q^{(i+1)}_{YZ|S}(y,z|s). 
 \label{eq:iter1}
\end{equation}

\paragraph{\underline{Stopping criterion}}
The iteration~\eqref{subeq:alternating} can be stopped when  
\begin{equation}
\max_{y\in\Ycal,z\in\Zcal}\log{\tfrac{Q^{(i+1)}_{YZ|S}(y,z|s)}{Q^{(i)}_{YZ|S}(y,z|s)}}\le \epsilon,\quad\text{for all }s\in\Scal, 
\label{eq:stopping}
\end{equation}
for some prescribed accuracy~$\epsilon>0$. When this condition is satisfied, $\epsilon$ is an upper bound on the difference of the current value of the divergence and the minimum~\citep[see][Corollary~5.1]{csiszar2004}. $\epsilon$ is a parameter of the algorithm.  In our experiments, we chose $\epsilon =10^{-6}$.

For the $I$-projection, the iteration~\eqref{eq:Iprojiteration} can be stopped when the squared distance between subsequent distributions is less than the square of some prespecified $\epsilon_1$. 
We found $\epsilon_1 = 10^{-2} \epsilon$ to be a good standard value. We call this the heuristic stopping criteria (Stop~1).

In general, the distributions returned in Step 1 are not exact, and this needs to be accounted for in the stopping criterion. 
In Appendix~\ref{sec:stopping}, we show that it is possible to guarantee $\epsilon$-optimality of the overall optimization, if the outer loop is interrupted when 
$\max_{y,z,s} \log \tilde Q^{(i+1)}_{YZ|S}(y,z|s) / \tilde Q^{(i)}_{YZ|S}(y,z|s) \leq \frac{\epsilon}{3}$, 
and iteration~\eqref{eq:Iprojiteration} is interrupted when $\|\tilde \eta^{(i)} - \eta\|_1 \leq \tilde Q^{(i)}_{YZ|S}(y,z|s) \frac{\epsilon}{12}$. Here $\tilde \eta^{(i)}$ and $\eta$ denote the expectation parameters of the current iterate and of the target, respectively. We call this the rigorous stopping criteria (Stop~2) .

\paragraph{Time complexity.}
In Appendix~\ref{sec:time_complexity}, we show that the overall time complexity of one iteration of Algorithm~\ref{alg:admUI} is dominated by Step~1. The time complexity of finding the $I$-projection (Algorithm~\ref{alg:Iproj}) depends on the distribution $b_0$ that is being projected. For a uniform distribution $b_0$, we show that the complexity of finding the $I$-projection to within~$\epsilon_1$ of the true solution is~$\mathcal{O}(\tfrac{|\Ycal||\Zcal|\log{\left(|\Ycal||\Zcal|\right)}}{\epsilon_1})$.

\paragraph{Modifications.} 
There are various natural modifications of our algorithms that can contribute to an improved performance. The stopping criterion does not need to be evaluated in every iteration. Evaluating it once every 20 iterations saved about 10\% of the total computation time, as we found in numerical experiments. For large systems, Step~1 can be run in parallel for blocks of $s$ values. The stopping criterion discussed previously will work regardless of the iterative optimization method used in Step~1. 

Step~1 is computationally demanding. So we expect that optimizing the computation of the $I$-projection will lead to further performance gains. 
In Appendix~\ref{sec:proximal}, we discuss a proximal point formulation of the GIS~\eqref{eq:Iprojiteration} that leads to the following iteration:
\begin{align}
b_{n+1}(y,z) = b_n(y,z) \bigg[\frac{P_{Y|S}(y|s)}{\sum_z b_n(y,z)}\bigg]^{1/(2\gamma)} \bigg[\frac{P_{Z|S}(z|s)}{\sum_y b_n(y,z)}\bigg]^{1/(2\gamma)}, 
\label{eq:Iprojiteration_accsoln}
\end{align}
for some~$\gamma\in \left(0, 1\right]$. We get back the GIS iteration~\eqref{eq:Iprojiteration} for~$~\gamma=1$.
Choosing$~\gamma<1$ has the potential for accelerating convergence. 
In Section~\ref{sec:experiments}, we report some preliminary findings using this approach.
A detailed analysis of the convergence properties of~\eqref{eq:Iprojiteration_accsoln} is reserved for future study.

\section{Experiments}
\label{sec:experiments}
\paragraph{Comparison with other methods.}
We compare the performance of our alternating divergence minimization algorithm admUI against other optimization methods. 
We implemented the \texttt{admUI} algorithm in Matlab R2017a as a Matlab executable (MEX). The Matlab source code as well as a standalone Python implementation are available online\footnote{\url{https://github.com/infodeco/computeUI}}.

Our first baseline is the general purpose optimizer \texttt{fmincon} (selecting the interior-point method) from the Matlab optimization package, with options for including the gradient and the Hessian.
The left panel in Figure~\ref{fig:1} shows the mean of the values of the unique information computed on $250$ joint distributions of $(S,Y,Z)$ sampled uniformly at random from the probability simplex. The right panel shows the average computation (wall-clock) time. We are interested in the accuracy of the computations and the required computation time as the state spaces increase in size. In terms of accuracy, all methods perform similarly (lower values reflect more accurate outcomes of the minimization). 
However, our algorithm allows for significant savings in terms of computation time. 
In fact, the black-box \texttt{fmincon} and \texttt{fmincon} with only the gradient included failed to give any answer beyond $|\Scal|=12$ in a reasonable amount of time (see last row in Figure~\ref{fig:1}). 

We also compared a Python implementation of the \texttt{admUI} algorithm with other software that solve~\eqref{eq:MI-decomposition}, viz. (a) an implementation using the conditional gradient method (also called the Frank-Wolfe algorithm) included in the Python package \texttt{dit}\footnote{R.~G. James, C.~J. Ellison, and J.~P. Crutchfield. \url{https://github.com/dit/dit}} (see also~\cite{pica2017quantifying}), and (b) an implementation in~\cite{DOT2017bivariate}\footnote{\url{https://github.com/Abzinger/BROJA-Bivariate-Partial_Information_Decomposition/blob/master/Python/cvxopt_solve.py}} using the Python interior-point solver \texttt{CVXOPT}\footnote{\url{cvxopt.org}}. 
Figure~\ref{fig:2} compares the mean of the values of the unique information computed on $100$ joint distributions of $(S,Y,Z)$ sampled uniformly from the simplex. 
The \texttt{dit} solver~(a) failed to converge a total of 14 times (out of 100) and took inordinately long to give any answer beyond $|\Scal|=5$. 
The performance of the \texttt{CVXOPT} implementation~(b) was comparable.

We note that for \texttt{admUI} we did not parallelize the computations in Step~1, which we expect will provide additional savings, especially for systems with large $\mathcal{S}$ (last row in the Figure~\ref{fig:1}). 
\begin{figure}
	\centering
	\subfloat{	
		\begin{tikzpicture}
		\begin{axis}[
		height=5cm,
		width=6cm,
		title={\textbf{$|\mathcal{S}|=|\mathcal{Z}|=2$}},
		xlabel=$|\mathcal{Y}|$,
		ylabel=$UI(S;Y\backslash Z)$,
		yticklabel style={
			/pgf/number format/fixed,
			/pgf/number format/precision=5},
		xmin=2,
		xmax=14,
		xtick={2,4,6,8,10,12,14},
		ymin=0
		]
		\addplot coordinates {
			(2,0.036337190330491)
			(3,0.069894355469846)
			(4,0.093929965681318)
			(5,0.100035935028438)
			(6,0.107657258125134)
			(7,0.108432820381229)
			(8,0.119815317628210)
			(9,0.124762783126674)
			(10,0.124475347984120)
			(12,0.124192837421809)
			(14,0.122183740651046)
		};\label{p1}
		
		\addplot coordinates {
			(2,0.036338033029132)
			(3,0.069894747628011)
			(4,0.093930455258234)
			(5,0.100036265396360)
			(6,0.107657606473771)
			(7,0.108433078432066)
			(8,0.119815448660499)
			(9,0.124762998685682)
			(10,0.124475525562697)
			(12,0.124230083871032)
			(14,0.122183833315786)
		};\label{p2}
		
		\addplot[green,mark=triangle*] coordinates {
			(2,	0.036337650118867)
			(3,	0.069894770437329)
			(4,	0.093930401251079)
			(5,	0.100036303527388)
			(6,	0.107657781266433)
			(7,	0.108433120038978)
			(8,	0.119815573022502)
			(9,	0.124763137657268)
			(10,0.124475629730962)
			(12,0.124260901254419)
			(14,0.122184203462197)
		};\label{p3}
		
		\addplot coordinates {
			(2, 0.036339393828858)
			(3, 0.069895377550413)
			(4, 0.093930913325256)
			(5, 0.100036657743370)
			(6, 0.107658432678804)
			(7, 0.108433536240814)
			(8, 0.119815620117481)
			(9, 0.122806497585793)
			(10,0.124625041842130)
			(12,0.126092490609513)
			(14,0.125091225869430)
		}; 
		\end{axis}
		\end{tikzpicture}}\qquad
	\subfloat{	
		\begin{tikzpicture}
		\begin{axis}[
		height=5cm,
		width=6cm,
		title={\textbf{$|\mathcal{S}|=|\mathcal{Z}|=2$}},
		xlabel=$|\mathcal{Y}|$,
		ylabel=time (s),
		yticklabel style={
			/pgf/number format/fixed,
			/pgf/number format/precision=5},
		xmin=2,
		xmax=14,
		xtick={2,4,6,8,10,12,14},
		ymin=0		
		]
		
		\addplot coordinates {
			(2, 0.010184840000000)
			(3, 0.018188912000000)
			(4,	0.020635012000000)
			(5,	0.025902420000000)
			(6,	0.055179072000000)
			(7,	0.033719548000000)
			(8,	0.019810860000000)
			(9,	0.060236128000000)
			(10,0.040368788000000)
			(12,0.048483542000000)
			(14,0.051733520000000)
		}; 
		
		\addplot coordinates {
			(2,	0.116160316000000)
			(3,	0.126403928000000)
			(4,	0.160154920000000)
			(5,	0.138765764000000)
			(6,	0.153746016000000)
			(7,	0.187359772000000)
			(8,	0.198938864000000)
			(9,	0.219070324000000)
			(10,0.222038420000000)
			(12,0.356783140000000)
			(14,0.388234590000000)
		}; 
		
		\addplot[green,mark=triangle*] coordinates {
			(2,	0.062509712000000)
			(3,	0.080577908000000)
			(4,	0.110016308000000)
			(5,	0.101113352000000)
			(6,	0.110677444000000)
			(7,	0.140571768000000)
			(8,	0.158299048000000)
			(9,	0.183299872000000)
			(10,0.201612412000000)
			(12,0.324351940000000)
			(14,0.380367040000000)	
		}; 

		\addplot coordinates {
			(2,	0.052614304000000)
			(3,	0.070538556000000)
			(4,	0.107255816000000)
			(5,	0.106469336000000)
			(6,	0.125625668000000)
			(7,	0.155867544000000)
			(8,	0.176508592000000)
			(9,	0.212513052000000)
			(10,0.245971740000000)
			(12,0.412341230000000)
			(14,0.548932090000000)	
		}; \label{p4}
		
		\end{axis}
		\end{tikzpicture}}
		
	\subfloat{
		\begin{tikzpicture}
		\begin{axis}[
		height=5cm,
		width=6cm,
		title={\textbf{$|\mathcal{S}|=|\mathcal{Y}|=2$}},
		xlabel=$|\mathcal{Z}|$,
		ylabel=$UI(S;Y\backslash Z)$,
		yticklabel style={
			/pgf/number format/fixed,
			/pgf/number format/precision=5},
		xmin=2,
		xmax=14,
		xtick={2,4,6,8,10,12,14},
		ymin=0,
		]
		\addplot coordinates {
			(2,0.045736439975042)   
			(3,0.018227198483102)   
			(4,0.011925900195795)   
			(5,0.006909394928942)   
			(6,0.004992134149047)   
			(7,0.002500562369503)   
			(8,0.001704971553163)   
			(9,0.000897618615722)   
			(10,0.001501752588882)
			(12,0.001213348602059) 
			(14,0.000417201504115)  
		}; 
		
		\addplot coordinates {
			(2,0.045737278128433)    
			(3,0.018233778872387)   
			(4,0.011926319031251)   
			(5,0.006989444670330)   
			(6,0.004992344767831)   
			(7,0.002500734144857)   
			(8,0.001705145506218)   
			(9,0.000897735615845)   
			(10,0.001502019817655) 
			(12,0.001213467965821)  
			(14,0.000417505596073)
		}; 
		
		\addplot[green,mark=triangle*] coordinates {
			(2,0.045736788964073)   
			(3,0.018227617289008)   
			(4,0.011926407562748)   
			(5,0.006909796909507)   
			(6,0.004992442765701)   
			(7,0.002500787471947)   
			(8,0.001705210202248)   
			(9,0.000897735567439)   
			(10,0.001502055890341)   
			(12,0.001213479642736)
			(14,0.000417505631941)
		}; 
		
		\addplot coordinates {
			(2,0.045737571973642)
			(3,0.018228024314051)
			(4,0.012003835016522)
			(5,0.006910061337917)
			(6,0.005023335778767)
			(7,0.002561441069351)
			(8,0.000541619186985)
			(9,0.001147150612226)
			(10,0.001576564546069)
			(12,0.002213709556396)
			(14,0.001975359919793)
		}; 
		\end{axis}
		\end{tikzpicture}}\qquad
	\subfloat{
		\begin{tikzpicture}
		\begin{axis}[
		height=5cm,
		width=6cm,
		title={\textbf{$|\mathcal{S}|=|\mathcal{Y}|=2$}},
		xlabel=$|\mathcal{Z}|$,
		ylabel=time (s),
		yticklabel style={
			/pgf/number format/fixed,
			/pgf/number format/precision=5},
		xmin=2,
		xmax=14,
		xtick={2,4,6,8,10,12,14},
		ymin=0,
		]
		\addplot coordinates {
			(2,0.018227548000000)   
			(3,0.018663604000000)   
			(4,0.031289084000000)   
			(5,0.049334788000000)   
			(6,0.025016444000000)   
			(7,0.018376484000000)   
			(8,0.065971608000000)   
			(9,0.020697672000000)   
			(10,0.041753772000000)  
			(12,0.042888100000000) 
			(14,0.058427660000000)
		};
		
		\addplot coordinates {
			(2,0.219638892000000)   
			(3,0.257696412000000)   
			(4,0.276246560000000)   
			(5,0.299752192000000)   
			(6,0.275850900000000)   
			(7,0.234101684000000)   
			(8,0.246174020000000)   
			(9,0.230280280000000)   
			(10,0.261678584000000)
			(12,0.419016620000000)  
			(14,0.394722020000000) 
		};
		
		\addplot[green,mark=triangle*] coordinates {
			(2,0.119682408000000)   
			(3,0.156066524000000)   
			(4,0.191287848000000)   
			(5,0.224432684000000)   
			(6,0.216187752000000)   
			(7,0.218527492000000)   
			(8,0.202126160000000)   
			(9,0.181216148000000)   
			(10,0.208251072000000)
			(12,0.356610620000000) 
			(14,0.422046400000000)
		};
		
		\addplot coordinates {
			(2,0.099716200000000)
			(3,0.144388624000000)
			(4,0.193335184000000)
			(5,0.254793576000000)
			(6,0.258407872000000)
			(7,0.244054416000000)
			(8,0.270803292000000)
			(9,0.260035188000000)
			(10,0.296725988000000)
			(12,0.4861233980000000)
			(14,0.5512800230000000)
		};
		\end{axis}
		\end{tikzpicture}}
	
	\subfloat{
		\begin{tikzpicture}
		\begin{axis}[
		height=5cm,
		width=6cm,
		title={\textbf{$|\mathcal{S}|=|\mathcal{Y}|=|\mathcal{Z}|$}},
		xlabel=$|\mathcal{S}|$,
		ylabel=$UI(S;Y\backslash Z)$,
		yticklabel style={
			/pgf/number format/fixed,
			/pgf/number format/precision=5},
		xmin=2,
		xmax=14,
		xtick={2,4,6,8,10,12,14},
		ymin=0,
		ymax=0.08
		]
		\addplot coordinates {
			(2,0.043512680387278)         
			(3,0.063660892006202)         
			(4,0.064301904820781)         
			(5,0.054014750668193)         
			(6,0.042740629455059)         
			(7,0.040607202116490)         
			(8,0.038388283890848)         
			(9,0.044589440874052)         
			(10,0.041653051338811)
			(12,0.036775540415215)
			(14,0.034750646409598)
		}; 
		
		\addplot coordinates {
			(2,0.043513514598897)
			(3,0.064661914222407) 
			(4,0.064304713906200) 
			(5,0.054023670966079) 
			(6,0.042762720228318) 
			(7,0.040639748898869) 
			(8,0.038440644159179) 
			(9,0.044668708136723) 
			(10,0.041766449996342) 
			(12,0.037329232191223)
			(14,0.035095490047353)
		}; 
		
		\addplot[green,mark=triangle*] coordinates {
			(2,0.043513052256636) 
			(3,0.063662249790459) 
			(4,0.064320893981959) 
			(5,0.054028836959206) 
			(6,0.042758365630448) 
			(7,0.040633360329410) 
			(8,0.038430776347114) 
			(9,0.044766692119743) 
			(10,0.042184649397181)  
			(12,nan)
			(14,nan)
		}; 
		
		\addplot coordinates {
			(2,0.043514471933175) 
			(3,0.063663431992091) 
			(4,0.064310864150548) 
			(5,0.054032406908049) 
			(6,0.042762361750877) 
			(7,0.040636071084364) 
			(8,0.038437724006580) 
			(9,0.044665761267988) 
			(10,0.041768220372454)
			(12,nan)
			(14,nan)
		}; 
		\end{axis}
		\end{tikzpicture}}\qquad
	\subfloat{
		\begin{tikzpicture}
		\begin{semilogyaxis}[
		log ticks with fixed point,
		height=5cm,
		width=6cm,
		title={\textbf{$|\mathcal{S}|=|\mathcal{Y}|=|\mathcal{Z}|$}},
		xlabel=$|\mathcal{S}|$,
		ylabel=time (s),
		yticklabel style={
			/pgf/number format/fixed,
			/pgf/number format/precision=5},
		xmin=2,
		xmax=14,
		xtick={2,4,6,8,10,12,14},
		]
		\addplot coordinates {
			(2,0.005935260000000)          
			(3,0.112118996000000)          
			(4,0.913945676000000)         
			(5,2.810476188000000)         
			(6,0.006247523500000*1000)         
			(7,0.007759633250000*1000)         
			(8,0.014050927750000*1000)         
			(9,0.017305351250000*1000)         
			(10,0.032366395750000*1000)
			(12,55.3027970000000) 
			(14,70.6605410000000)
		};
		
		\addplot coordinates {
			(2,0.114889824000000) 
			(3,0.259117848000000)  
			(4,0.613103511999999)  
			(5,1.203194124000001)  
			(6,0.002973777500000*1000) 
			(7,0.006810948500000*1000) 
			(8,0.012776858000000*1000) 
			(9,0.022315438750000*1000) 
			(10,0.053615616500000*1000)
			(12,134.0371640000000)
			(14,360.2420370000000)
		};
		
		\addplot[green,mark=triangle*] coordinates {
			(2,0.057118648000000)  
			(3,0.179872432000000)  
			(4,0.611441592000000)  
			(5,2.752290239999999)  
			(6,0.012235002500000*1000) 
			(7,0.042256073250000*1000) 
			(8,0.121243308000000*1000) 
			(9,0.238880543250000*1000) 
			(10,1.909168042750000*1000)
			(12,nan)
			(14,nan)
		};
		
		\addplot coordinates {
			(2,0.055964796000000)  
			(3,0.276478288000000) 
			(4,1.288338032000000) 
			(5,5.587155328000001) 
			(6,0.022046305000000*1000) 
			(7,0.075182621500000*1000) 
			(8,0.207125876750000*1000) 
			(9,0.620672913500000*1000) 
			(10,1.461862588250000*1000)
			(12,nan)
			(14,nan)
		};
		\end{semilogyaxis}
		\end{tikzpicture}}
	\caption{Comparison of a Matlab implementation of our \texttt{admUI} algorithm (\ref{p1}) with \texttt{fmincon} from the Matlab Optimization Toolbox (algorithm: interior-point) when including the gradient and Hessian (\ref{p2}), only the gradient (\ref{p3}), and when including none (\ref{p4}). The left panel shows the average values of the computed unique information for 250 distributions sampled uniformly at random from the probability simplex. The right panel shows the average computation (wall-clock) time on an Intel 2.60~GHz CPU. Note the semilog ordinate in the right panel of the last row which corresponds to much larger systems. 
	}
	\label{fig:1}
\end{figure}
\begin{figure}[tb]
	\centering
	\subfloat{
		\begin{tikzpicture}
		\begin{axis}[
		height=5cm,
		width=6cm,
		title={\textbf{$|\mathcal{S}|=|\mathcal{Y}|=|\mathcal{Z}|$}},
		xlabel=$|\mathcal{S}|$,
		ylabel=$UI(S;Y\backslash Z)$,
		yticklabel style={
			/pgf/number format/fixed,
			/pgf/number format/precision=5},
		xmin=2,
		xmax=12,
		xtick={2,4,6,8,10,12},
		ymin=0,
		ymax=0.08
		]
		\addplot coordinates {
			(2,0.045447558031619)         
			(3,0.069316876662204)         
			(4,0.069043696136647)         
			(5,0.050866104974674)         
			(6,0.050661574574079)         
			(7,0.045957348534485)         
			(8,0.041593214463482)         
			(9,0.040746277625364)         
			(10,0.037862119198596)
			(11,0.034760498685571)
			(12,0.033256839016617)
		};  \label{p8}
		
		\addplot coordinates {
			(2,0.045496607982562)
			(3,0.06956810569746) 
			(4,0.069051447599844) 
			(5,0.049985894585717) 
		};  \label{p9}
		
		\addplot coordinates {
			(2,0.045447790929084) 
			(3,0.069316248963056) 
			(4,0.069037961879531) 
			(5,0.050855452902266) 
			(6,0.050651441101571) 
			(7,0.045943370993516) 
			(8,0.041583073309663) 
			(9,0.040732305256225) 
			(10,0.037848576630177)  
			(11,0.034748576568929)
			(12,0.033304105596363)
		};  \label{p10}
		\end{axis}
		\end{tikzpicture}}\qquad
	\subfloat{
		\begin{tikzpicture}
		\begin{semilogyaxis}[
		log ticks with fixed point,
		height=5cm,
		width=6cm,
		title={\textbf{$|\mathcal{S}|=|\mathcal{Y}|=|\mathcal{Z}|$}},
		xlabel=$|\mathcal{S}|$,
		ylabel=time (s),
		yticklabel style={
			/pgf/number format/fixed,
			/pgf/number format/precision=5},
		xmin=2,
		xmax=12,
		xtick={2,4,6,8,10,12},
		ymax=28
		]
		\addplot coordinates {
			(2,0.286711366176605)          
			(3,1.323447477817535)          
			(4,2.889853680133819)         
			(5,3.950668849945068)         
			(6,4.890721967220307)         
			(7,5.464065339565277)         
			(8,6.205673131942749)         
			(9,6.895060749053955)         
			(10,7.582759811878204)
			(11,7.479078481197357)
			(12,8.120807814598084)
		};
		
		\addplot coordinates {
			(2,0.152317404747009) 
			(3,1.018656647205353)  
			(4,5.964616098403931)  
			(5,18.01919047355652)  
		};
		
		\addplot coordinates {
			(2,0.039378819465637)  
			(3,0.1547265625)  
			(4,0.369298918247223)  
			(5,0.703953814506531)  
			(6,1.266826131343842) 
			(7,2.112130069732666) 
			(8,3.659219310283661) 
			(9,5.795591287612915) 
			(10,9.383078887462617)
			(11,14.74226320743561)
			(12,24.89692880392074)
		};
		\end{semilogyaxis}
		\end{tikzpicture}}
	\caption{Comparison of a Python implementation of our \texttt{admUI} algorithm (\ref{p8}) with (a) the Frank-Wolfe algorithm in the Python package \texttt{dit}(\ref{p9}), and (b) a custom implementation in~\cite{DOT2017bivariate} (\ref{p10}) using the Python interior-point solver \texttt{CVXOPT}. The left panel shows the average values of the computed unique information, $UI(S;Y\setminus Z)$ for 100 distributions sampled uniformly at random from the probability simplex. The right panel shows the average computation (wall-clock) time on an Intel 2.60~GHz CPU. 
	}
	\label{fig:2}
\end{figure}
\begin{figure}[tb]
	\centering
	\begin{tikzpicture}
	\begin{axis}[
		height=5cm,
		width=8cm,
		title={\textbf{$|\mathcal{S}|=|\mathcal{Y}|=|\mathcal{Z}|=2$}},
		xlabel=Iterations,
		ylabel=$\text{log}_{10} \epsilon$,
		ymin=-12,
		ymax=-1,
		ytick={-12,-10,-8,-6,-4,-2,0},
		xmin=10,
		xmax=210,
		xtick={1,10,100,200}
		]
		\addplot coordinates {
			(2.084453383398952*100,-12)
			(1.926598846014098*100,-11)
			(1.732213561992118*100,-10)
			(1.539013767947880*100,-9)
			(1.387289736651987*100,-8)
			(1.194349423137010*100,-7)
			(1.004927344611121*100,-6)
			(78.247736093340819,-5)
			(56.343582270727950,-4)
			(36.003940667152477,-3)
			(20.257656193866595,-2)
			(11.491420634920631,-1)
		};  \label{p11}
		
		\addplot coordinates {
			(1.453308108308998*100,-12)
			(1.323805970314031*100,-11)
			(1.253765597267874*100,-10)
			(1.112784015024520*100,-9)
			(0.971917358984814*100,-8)
			(0.852615861947338*100,-7)
			(0.701769328620950*100,-6)
			(0.549715123671325*100,-5)
			(0.399822235808159*100,-4)
			(0.259931908333099*100,-3)
			(0.151507707453673*100,-2)
			(0.091649682539683*100,-1)
		};  \label{p12}
		
		\end{axis}
		\end{tikzpicture}
		\caption{Convergence of the \texttt{admUI} algorithm using the modified GIS~\eqref{eq:Iprojiteration_accsoln} with $\gamma=1/\sqrt{2}$ (\ref{p12}) and the original GIS~\eqref{eq:Iprojiteration} (\ref{p11}). For a prescribed accuracy~$\epsilon$, each point corresponds to the average number of iterations required by either method to converge for 250 distributions sampled uniformly at random from the probability simplex. 
		}
		\label{fig:3}
\end{figure}
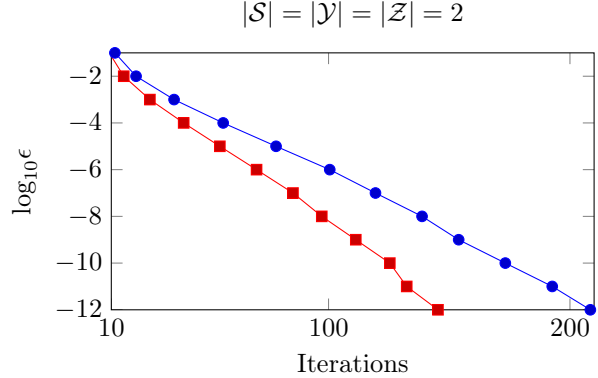
\paragraph{Accuracy and stopping criterion.}
To test the accuracy and efficiency of the \texttt{admUI} algorithm for high-dimensional systems, we consider the \textsc{Copy} distribution: $Y$ and $Z$ are independent uniformly distributed random variables and $S=(Y,Z)$. In this case, $\Delta_P=\{P\}$,
and $UI(S;Y\setminus Z)$ is just the mutual information $I(S;Y)=H(Y)$, which can be calculated exactly, given $P$. We use this example to test the accuracy of the solutions produced by different optimizers. Table~\ref{tab:comp_admi_fmc} compares the \texttt{admUI} algorithm and \texttt{fmincon} (with gradient and Hessian included) in terms of the error and computation times for different cardinalities of $\Ycal$. We chose $\Zcal=\Ycal$ and $\Scal=\Ycal\times\Ycal$ so that overall size of the system scales as $|\Ycal|^4$. Compared to the \texttt{admUI}, the computation time and error grow at a much faster rate for \texttt{fmincon}. 
\begin{table}[!htbp]
	\centering 
	\begin{threeparttable}
		\captionsetup{justification=centering}
		\caption{Comparison of \texttt{admUI} and \texttt{fmincon} on the \textsc{Copy} example. 
		} 
		\label{tab:comp_admi_fmc} 
		\small 
		\begin{tabular}{@{\extracolsep{1pt}} cclclccc} 
			\toprule
			\multicolumn{1}{c}{Size} & \multicolumn{1}{c}{$\epsilon$} & \multicolumn{4}{c}{\texttt{admUI}} & \multicolumn{2}{c}{\texttt{fmincon}\tnote{1}}\\
			\cmidrule{3-6} \cmidrule{7-8}
			& & \multicolumn{2}{c}{Stop 1 (heuristic)}  & \multicolumn{2}{c}{Stop 2 (rigorous)}\\
			\cmidrule{3-4} \cmidrule{5-6}
			& & Error & Time (ms) & Error & Time (ms) & Error & Time (ms)\\
			\midrule  
			$\phantom{1}2^4$ & \num{e-08} & $1.94\cdot\num{e-9}$  &  \phantom{1}4.38   & $9.16\cdot\num{e-10}$ &  $9.03\cdot\num{e1}$ & $9.52\cdot\num{e-5}$ & $2.38\cdot\num{e2}$\\
			$\phantom{1}$    & \num{e-05} & $1.97\cdot\num{e-6}$     &  \phantom{1}5.36   & $6.67\cdot\num{e-7}$ &  $6.45\cdot\num{e1}$ & $\phantom{1}$ & $\phantom{1}$\\  
			$\phantom{1}$    & \num{e-03} & $1.09\cdot\num{e-4}$   &  \phantom{1}4.19   & $5.01\cdot\num{e-5}$            &  $1.03\cdot\num{e1}$ & $\phantom{1}$ & $\phantom{1}$\\  
			\midrule  
			$\phantom{1}4^4$ & \num{e-08} & $1.63\cdot\num{e-9}$  & 11.09   & $7.24\cdot\num{e-10}$ &  $2.27\cdot\num{e2}$ & $1.50\cdot\num{e-4}$ & $4.17\cdot\num{e2}$\\
			$\phantom{1}$    & \num{e-05} & $1.84\cdot\num{e-6}$     &  \phantom{1}5.77   & $5.38\cdot\num{e-7}$ &  $2.67\cdot\num{e2}$ & $\phantom{1}$ & $\phantom{1}$\\  
			$\phantom{1}$    & \num{e-03} & $1.03\cdot\num{e-4}$   &  \phantom{1}5.06   & $4.13\cdot\num{e-5}$            &  $2.59\cdot\num{e2}$ & $\phantom{1}$ & $\phantom{1}$\\
			\midrule  
			$\phantom{1}7^4$ & \num{e-08} & $3.15\cdot\num{e-9}$  &  \phantom{1}6.23   & $4.93\cdot\num{e-10}$ &  $2.42\cdot\num{e3}$ & $2.32\cdot\num{e-4}$ & $8.61\cdot\num{e3}$\\
			$\phantom{1}$    & \num{e-05} & $1.43\cdot\num{e-6}$     &  \phantom{1}4.49   & $3.71\cdot\num{e-7}$ &  $2.41\cdot\num{e3}$ & $\phantom{1}$ & $\phantom{1}$\\  
			$\phantom{1}$    & \num{e-03} & $0.81\cdot\num{e-4}$   &  \phantom{1}7.68   & $2.89\cdot\num{e-5}$            &  $1.97\cdot\num{e3}$ & $\phantom{1}$ & $\phantom{1}$\\ 
			\midrule  
			$\phantom{1}10^4$& \num{e-08} & $2.60\cdot\num{e-9}$  & 14.67   & $3.71\cdot\num{e-10}$ &  $9.38\cdot\num{e3}$ & $3.51\cdot\num{e-4}$ & $4.86\cdot\num{e5}$\\
			$\phantom{1}$    & \num{e-05} & $1.18\cdot\num{e-6}$     & 12.11   & $2.82\cdot\num{e-7}$ &  $9.20\cdot\num{e3}$ & $\phantom{1}$ & $\phantom{1}$\\  
			$\phantom{1}$    & \num{e-03} & $0.66\cdot\num{e-4}$   & 11.90   & $2.22\cdot\num{e-5}$            &  $8.73\cdot\num{e3}$ & $\phantom{1}$ & $\phantom{1}$\\
			\bottomrule
		\end{tabular} 
		\begin{tablenotes}
			\item[1] \texttt{fmincon} with gradient, Hessian, and options: Algorithm = interior-point, MaxIterations $=10^4$, MaxFunctionEvaluations $=10^5$, OptimalityTolerance = $\num{e-6}$, ConstraintTolerance = $\num{e-8}$. 
		\end{tablenotes}
	\end{threeparttable}
\end{table} 	 

For \texttt{admUI}, we consider the two stopping criteria discussed in Section~\ref{sec:computing}, with several choices of the accuracy parameter $\epsilon$. Stop~1 is the heuristic and Stop~2 is the rigorous method. The stopping criterion was evaluated in every iteration. As can be seen from the table, both criteria allow us to control the error. The heuristic has a lower computational overhead compared to the rigorous stopping criterion. On the other hand, the error bound of the rigorous criterion appears to be somewhat pessimistic, and seems to perform well even with a much larger $\epsilon$. 

\paragraph{Accelerating the $I$-projection (Step 1).} 
For 250 distributions sampled uniformly at random from the probability simplex, Figure~\ref{fig:3} compares the mean number of iterations required by the \texttt{admUI} algorithm when using the original GIS~\eqref{eq:Iprojiteration} and when using the modified  GIS~\eqref{eq:Iprojiteration_accsoln} with $\gamma=1/\sqrt{2}$ to achieve a given accuracy~$\epsilon$. The convergence of the \texttt{admUI} algorithm with the modified GIS is noticeably faster when compared to the original.

\section{Discussion}
\label{sec:discussion}
We developed an efficient algorithm to compute the decomposition of mutual information proposed in~\cite{e16042161}, 
for which the computation had remained a challenge so far. Our algorithm comes with convergence guarantees and a rigorous stopping criterion ensuring $\epsilon$-optimality of the solution. 
We tested the computation time and accuracy of our algorithm against other software. 
In a number of experiments, our algorithm is shown to perform more accurately and efficiently than previous approaches. 

One may ask whether the computational complexity of the function $UI$ prohibits its use in applications, given that already computing or estimating a mutual information is challenging. One major problem when estimating the mutual information is the difficulty in estimating the joint distribution of many variables. In this respect, $UI$ compares well, since $UI(S;Y\backslash Z)$ does not depend on the joint distribution of all variables, but only on the marginal distributions of pairs $(S,Y)$ and~$(S,Z)$. In those applications where the main problem is the estimation of the joint distribution given the data at hand, $UI$ is easier to treat than the mutual information. 

We hope that our algorithm will contribute means to test the mutual information decomposition on larger systems than was possible so far, particularly in recent applications of the decomposition, e.g., in neuroscience~\citep{pica2017quantifying}, representation learning~\citep{minsyn2017,e19090474}, robotics~\citep{ghazi2015quantifying,e19090456}, etc., which so far has been pursued either with only simpler types of measures or for very low-dimensional systems.

\appendix
\section*{Appendix}

\section{Stopping criterion}
\label{sec:stopping}
\paragraph{Outer loop with errors.}
The stopping criterion~\eqref{eq:stopping} for the outer loop of Algorithm~\ref{alg:admUI} tests $\max_{s,y,z}\log\frac{Q^{(i+1)}(y,z|s) }{Q^{(i)}(y,z|s)}\leq \epsilon$, which ensures that the objective function has reached a value within $\epsilon$ of optimal. 
We need to describe the behavior of this test when using approximations $\tilde Q^{(i)}$ and $\tilde Q^{(i+1)}$ instead of the exact distributions $Q^{(i)}$ and $Q^{(i+1)}$. 
Consider any $s,y,z$ and abbreviate $q^{(i)}\equiv Q^{(i)}(y,z|s)$ and $\tilde q^{(i)}\equiv \tilde Q^{(i)}(y,z|s)$. 
\begin{proposition}
	Let $\epsilon>0$. 
If 
	\begin{equation*}
	|\tilde q^{(i)} - q^{(i)}| \leq  \tilde q^{(i)} \frac{\epsilon}{12},  \quad 
	|\tilde q^{(i+1)} - q^{(i+1)}| \leq  \tilde q^{(i+1)} \frac{\epsilon}{12},  
\quad\text{and}\quad
	\log\frac{\tilde q^{(i+1)}}{\tilde q^{(i)}}\leq 
	\frac{\epsilon}{3}, 
	\end{equation*}
then 
	\begin{equation*}
	\log\frac{q^{(i+1)}}{q^{(i)}}\leq \epsilon. 
	\end{equation*}
\end{proposition}
\begin{proof}
	By direct evaluation. 
\end{proof}
In turn, testing the stopping criterion with $\epsilon_1\leq \frac{\epsilon}{3}$ allows us to conclude $\epsilon$-optimality, if the approximate distributions plugged in are within $\epsilon_0\leq \min\{ \tilde q^{(i)}, \tilde q^{(i+1)}\}\frac{\epsilon}{12}$ of the actual distributions, in each entry. 

\paragraph{Inner loop.}
Now we want to find a criterion to interrupt the iteration from Algorithm~\ref{alg:Iproj} with the guarantee that $|\tilde q - q|\leq \epsilon_0$ for some prespecified $\epsilon_0$. 

Note that the optimization in Theorem~\ref{theorem:innerloop} takes place over the set of distributions of the form $\frac{1}{Z(R_{YZ},q_Y,q_Z)}R_{YZ}(y,z)q_Y(y)q_Z(z)$, where $q_Y$ and $q_Z$ are arbitrary probability distributions over $Y$ and $Z$ respectively, $R$ is the distribution that we want to approximate with a distribution from the linear family $\Delta_{P,s}$, and $Z(R_{YZ},q_Y,q_Z)$ is the normalizing partition function. This is an exponential family with sufficient statistics $\mathds{1}_{y'}$, $y'\in Y$, $\mathds{1}_{z'}$, $z'\in Z$, computing the marginal distributions on $Y$ and $Z$. 
(This is similar to an independence model, but with a non uniform reference measure.) The solution to this optimization problem is the unique distribution $Q_{YZ}$ within the exponential family, that is also contained in $\Delta_{P,s}$, meaning that its marginal distributions (which correspond to the expectation parameters) satisfy $Q_Y(y) = \eta_y = P_{Y|S}(y|s)$ and $Q_Z(z)=\eta_z =P_{Z|S}(z|s)$. We want to bound the error $|\tilde q - q|$ in terms of the error $|\tilde \eta - \eta|$ of the expectation parameters. 

\begin{conjecture}
	\label{conj}
	$\|\tilde q - q\|_\infty \leq  \|\tilde \eta - \eta\|_1$. 
\end{conjecture}

Extensive computer experiments seem to confirm that Conjecture~\ref{conj} is true. 
Assuming this, the stopping criterion is 
\begin{align*}
\|\tilde \eta - \eta\|_1 =
& \sum_{z\in Z\setminus \{1\}} |(\sum_y \frac{1}{Z}R_{YZ}(y,z)q_Y(y)q_Z(z)) - P_{Z|S}(z|s)| \\
&+ \sum_{y\in Y\setminus \{1\}}|(\sum_z \frac{1}{Z}R_{YZ}(y,z)q_Y(y)q_Z(z)) - P_{Y|S}(z|s)|\\
\leq&  \epsilon_0. 
\end{align*}

Summarizing, we can guarantee $\epsilon$-optimality of the overall optimization, if the outer loop is interrupted when $\log\frac{\tilde q^{(i+1)} }{\tilde q^{(i)}}\leq 
 \frac{\epsilon}{3}$, 
and the inner loop is interrupted when 
$\frac{\|\tilde \eta^{(i)} - \eta\|_1 }{ \min \tilde{q}^{(i)} } \leq \frac{\epsilon}{12}$. 

\section{Time complexity}
\label{sec:time_complexity}
The convergence analysis of the GIS is similar to that of the classical BAA~\citep{1054855,arimoto1972} (see below). We have the following proposition.
\begin{proposition}
\label{prop:time_complexity}
Let~$b_n$ be the nonnegative functions defined recursively on~$\Ycal\times\Zcal$ by~\eqref{eq:Iprojiteration} and let~$b^*(j)=\lim_{n\to\infty} b_n(j),\text{ }j\in\Ycal\times\Zcal$. Then the approximation error~$\lvert D(b_{n+1}\|b_0)-D(b^*\|b_0)\rvert$ is upper-bounded by~$\tfrac{D(b^*\|b_0)}{n}$. In particular, if~$b_0$ is the uniform distribution, then the error bound is of the form~$\mathcal{O}(\tfrac{\log{\left(|\Ycal||\Zcal|\right)}}{n})$.
\end{proposition}
\begin{proof}
See~\citet[Section 5]{SullivanAM} and~\citet[Corollary 1, p.~17]{arimoto1972}.	
\end{proof}
The time complexity analysis for the $I$-projection is not uniform, with the error bound $\tfrac{D(b^*\|b_0)}{n}$ depending on the distribution~$b_0$ that is being projected. Each iteration~\eqref{eq:Iprojiteration} costs~$\mathcal{O}(|\Ycal||\Zcal|)$ operations.
Hence by Proposition~\ref{prop:time_complexity}, the time complexity of finding the $I$-projection (Algorithm~\ref{alg:Iproj}) to within~$\epsilon_1$ of the true solution is~$\mathcal{O}(\tfrac{|\Ycal||\Zcal|\log{\left(|\Ycal||\Zcal|\right)}}{\epsilon_1})$. 
The $I$-projection needs to be evaluated for each~$s\in\Scal$ in Step~1 in Algorithm~\ref{alg:admUI}.
Hence, the complexity of one regular iteration of Step~1 is about~$\mathcal{O}(\tfrac{|\Scal||\Ycal||\Zcal|\log{\left(|\Ycal||\Zcal|\right)}}{\epsilon_1})$.  
The complexity of Step~2 is~$\mathcal{O}(|\Scal||\Ycal||\Zcal|)$ so that the overall complexity of one iteration of Algorithm~\ref{alg:admUI} is dominated by Step~1.

\paragraph{GIS and the classical BAA.} It is instructive to compare the time complexity of the GIS with that of the classical BAA. 
Let $J=|\Ycal||\Zcal|$,  $K=|\Ycal|+|\Zcal|$. 
Given a value~$s\in\Scal$, the set~$\Delta_{P,s}$~\eqref{eq:delPs} is the fiber of the linear map~$f_A: \mathbb{P}_{\Ycal\times\Zcal} \to \mathbb{P}_{\Ycal}\times\mathbb{P}_{\Zcal}$ passing through a given point~$P_{YZ|s}$.
Write~$A$ as the column-stochastic matrix~$(a_{kj})_{k,j}\in\mathbb{R}_+^{K\times J}$ describing~$f_A$ and let~$\alpha=AP_{YZ|s}\in\mathbb{R}_+^{K}$. Then~$\Delta_{P,s}=(P_{YZ|s}+\ker{A})\cap\mathbb{P}_{\Ycal\times\Zcal}=\{b\in\mathbb{P}_{\Ycal\times\Zcal}:Ab=\alpha\}$. As a concrete example, for~$\Ycal=\{0,1\}$, $\Zcal=\{0,1,2\}$, $A$ is the~$5\times 6$ matrix 
$$A=\begin{blockarray}{cccccc}
\matindex{00} & \matindex{01} & \matindex{02} & \matindex{10} & \matindex{11} & \matindex{12} & \\
\begin{block}{(cccccc)}
.5 & .5 & .5 & 0 & 0 & 0\\
0 &  0 &  0 &.5 &.5 &.5\\
.5 &  0 &  0 &.5 & 0 & 0\\
0 & .5 &  0 & 0 &.5 & 0\\
0 &  0 & .5 & 0 & 0 &.5\\
\end{block}
\end{blockarray}.$$
Rewrite~\eqref{eq:Iprojiteration} as
\begin{align}
b_0(j) &= R_{YZ}(j),\notag\\
b_{n+1}(j) &= b_n(j) \prod_{k}\bigg(\frac{\alpha_k}{\sum_{j'} a_{kj'}b_n(j')}\bigg)^{a_{kj}}
= b_n(j) \prod_{k}\bigg(\frac{a_{kj}}{\alpha_k}\bigg)^{-a_{kj}}\prod_{k}\bigg(\frac{a_{kj}}{\sum_{j'} a_{kj'}b_n(j')}\bigg)^{a_{kj}}\notag\\
&= b_n(j)\exp{(-D(A_j\|\alpha))}\prod_{k}\bigg(\frac{a_{kj}}{\sum_{j'} a_{kj'}b_n(j')}\bigg)^{a_{kj}},
\label{eq:BACC_link}
\end{align}
where $A_j$ is the $j$th column of $A$. If for all $j$ such that $b(j)>0$, $D(A_j\|\alpha)$ is a constant, then~\eqref{eq:BACC_link} reduces to the classical BAA~\citep{1054855,arimoto1972} for computing the capacity of a channel with transition matrix~$A$.
This is related to the Kuhn-Tucker conditions for finding the capacity-achieving input distribution: A distribution $b^*$ achieves capacity if and only if there exists a number $\mathcal{C}$ such that $D(A_j\|Ab^*)=\mathcal{C}$ for all $j$ such that $b^*(j)>0$ and $D(A_j\|Ab^*)<\mathcal{C}$ for all other $j$~\citep{SullivanAM}. The number~$\mathcal{C}$ is the channel capacity.
The standard BAA iteration~\citep{1054855,arimoto1972} for a channel with transition matrix~$A=(a_{kj})_{k,j}\in\mathbb{R}_+^{K\times J}$ is
\begin{align}
b_{n+1}(j) &= b_n(j) \prod_{k}\bigg(\frac{a_{kj}}{\sum_{j'} a_{kj'}b_n(j')}\bigg)^{a_{kj}},
\label{eq:BACC}
\end{align}
where $b_0$ is chosen from the interior of $\mathbb{P}_{\Ycal\times\Zcal}$. Each iteration~\eqref{eq:BACC_link} costs~$\mathcal{O}(KJ)$ operations. The approximation error scales inversely with the number of iterations~$n$ and is of the form~$\mathcal{O}(\tfrac{\log{J}}{n})$~\citep[Corollary 1, p.~17]{arimoto1972} so that the time complexity of finding the capacity to within~$\epsilon$ of the true solution is~$\mathcal{O}(KJ\tfrac{\log{J}}{\epsilon})$. 

Compare this with the GIS which achieves $\epsilon$-optimality with a complexity of~$\mathcal{O}(J\tfrac{\log{J}}{\epsilon})$. This can be attributed to the special structure of the matrix $A$ (see example above) that corresponds to the fiber polytope of the independence model.

\section{Proximal point formulation of the I-projection step}
\label{sec:proximal}
The GIS iteration~\eqref{eq:Iprojiteration} in Step 1 can be written as (see, e.g.,~\cite{SullivanAM})
\begin{align}
b_{n+1}(y,z) = \argmax_{b(y,z)} \Bigg\{\sum_{y,z}&\frac{b(y,z)}{2}\log\frac{P_{Y|S}(y|s)P_{Z|S}(z|s)}{\sum_z b_n(y,z)\sum_y b_n(y,z)}-\gamma D(b(y,z)\|b_n(y,z))\Bigg\},
\label{eq:Iprojiteration_acc}
\end{align}
with the tuning parameter $\gamma$ set to 1. The $D(b(y,z)\|b_n(y,z))$ term can be interpreted as a regularization term that penalizes updates $b_{n+1}$ that stray away from the vicinity of $b_n$. Selecting $\gamma<1$ has the potential for accelerating convergence. Techniques of this type are referred to as proximal point algorithms. Such a proximal point formulation with an adaptive tuning parameter~$\gamma$ has been used, for instance, to accelerate the convergence of the classical Blahut-Arimoto algorithm for computing the channel capacity~\cite{matz2004}. 
Using the method of Lagrange multipliers, it is easy to show that the solution to~\eqref{eq:Iprojiteration_acc} is given by~\eqref{eq:Iprojiteration_accsoln}.
In Section~\ref{sec:experiments}, we report some preliminary findings on accelerated convergence of Algorithm~\ref{alg:Iproj} using a value of the tuning parameter $\gamma < 1$ in~\eqref{eq:Iprojiteration_accsoln}. 

\section{Properties of the Unique information}
\label{sec:properties}
The decomposition of the mutual information~\eqref{eq:MI-decomposition} is motivated by an operational interpretation of the unique information. We briefly discuss the properties of the definition~\eqref{eq:uidefinition}.

Intuitively, if $Y$ has some unique information about $S$ (that is not known to $Z$), then there must be some way to exploit this information. Conversely, if $Z$ knows everything that $Y$ knows about $S$, then $Y$ can have no unique information about $S$. The following property formalizes this intuition~\cite[Lemma 6]{e16042161}.
\begin{itemize}
	\item[\textbf{(P1)}] Given $(S,Y,Z)\sim P$, $UI(S;Y\setminus Z)$ vanishes if and only if there exists a random variable $Y'$ such that the pairs $(S,Y)$ and $(S,Y')$ have the same distribution, and $S-Z-Y'$ is a Markov chain.
\end{itemize}
Blackwell’s theorem~\citep{Blackwell1953} implies that this property is equivalent to the fact that any decision problem in which the objective is to predict $S$ can be solved just as well with the knowledge of $Z$ as with the knowledge of $Y$. 
See~\cite{e16042161,BlackwellISIT} for a more detailed discussion. \textbf{(P1)} depends only on the channels $P_{Y|S}$ and $P_{Z|S}$ and thus on the marginals $P_{SY}$ and $P_{SZ}$.
One can also argue that any measure of unique information should depend only on $P_{SY}$ and $P_{SZ}$, but not on the full joint $P$. This is satisfied by the function~$UI$ since,
\begin{itemize}
	\item[\textbf{(P2)}] The functions $Q\mapsto UI_{Q}(S;Y\setminus Z)$ and $Q\mapsto UI_{Q}(S;Z\setminus Y)$ are constant in $\Delta_{P}$.
\end{itemize}
Since $UI$ satisfies \textbf{(P2)}, so does the function~$SI$. This follows from~\eqref{eq:MI-decomposition-2}. Only the function~$CI$ depends on the full joint~$P$. 

Like the mutual information, $UI$, $SI$, and $CI$ also satisfy an intuitive additivity property when evaluated on i.i.d. repetitions~\cite[Lemma 19]{e16042161}.
\begin{itemize}
	\item[\textbf{(P3)}] For sequences~$S^n=(S_1,\ldots,S_n)$,~$Y^n=(Y_1,\ldots,Y_n)$,~$Z^n=(Z_1,\ldots,Z_n)$ drawn i.i.d~$\sim P$, we have,~$UI(S^n;Y^n\setminus Z^n)=nUI(S;Y\setminus Z)$,~$UI(S^n;Z^n\setminus Y^n)=nUI(S;Z\setminus Y)$,~$SI(S^n;Y^n,Z^n)=nSI(S;Y,Z)$,~$CI(S^n;Y^n,Z^n)=nCI(S;Y,Z)$.   
\end{itemize}

The following properties are analogous to the data processing inequalities for the mutual information:
\begin{itemize}    
	\item[\textbf{(P4)}] Monotonicity under coarse-graining: Let $S'$, $Y'$ and $Z'$ be functions of $S$, $Y$, and $Z$, respectively. Then, the following data processing-like inequalities hold:
	\begin{itemize}
		\item $UI(S;Y\setminus Z) \ge UI(S';Y\setminus Z)$, 
		\item $UI(S;Y\setminus Z) \ge UI(S;Y'\setminus Z)$, 
		\item $UI(S;Y\setminus Z) \le UI(S;Y\setminus Z')$.
	\end{itemize}
\end{itemize}
See~\cite{e16042161,ISIT_RBOJ14} for further properties. 

Other nonnegative decompositions proposed so far, notably the information-geometric approach in~\cite{HarderSalgePolani2013:Bivariate_redundancy} and the approach in~\cite{GriffithKoch2014:Quantifying_Synergistic_MI} satisfy~\textbf{(P1)} and~\textbf{(P2)}.~\textbf{(P1)} is satisfied only by the decompositions in~\cite{HarderSalgePolani2013:Bivariate_redundancy,GriffithKoch2014:Quantifying_Synergistic_MI}. 
Notably,~\cite{HarderSalgePolani2013:Bivariate_redundancy,WilliamsBeer} do not satisfy~\textbf{(P3)}.  

\section{Notation}
\label{sec:notation}
We use capital letters to denote random variables and script for the corresponding finite alphabets. We write $P_S$ for the probability distribution of $S$, which is a vector with entries  $P_S(s)$, $s\in\mathcal{S}$. 
The support of $P_S$ is the set $\supp(P_S)=\{s\in\Scal:P_S(s)\neq 0\}$. 
The set of all probability measures on $\Scal$ is denoted $\mathbb{P}_{\Scal}$.  
A transition probability kernel from $\Scal$ to $\Ycal$ is a measurable function $P_{Y|S}:\Scal \to \mathbb{P}_{\Ycal}$, represented by a matrix with columns $P_{Y|S=s}=P_{Y|s}\in\mathbb{P}_\Ycal$, $s\in\Scal$. 

We use the following quantities: 
\begin{itemize}
	\item The entropy $H(P_S)$ of a distribution $P_S\in \mathbb{P}_{\Scal}$ is $H(P_S)=-\sum_{s\in\Scal}P_S(s)\log{P_S(s)}$.
    \item Given $P_S$,$Q_S\in\mathbb{P}_{\Scal}$, the Kullback-Leibler divergence from $P_S$ to $Q_S$ is $D(P_S\|Q_S)=\sum_{s\in \Scal} P_S(s)\log\tfrac{P_S(s)}{Q_S(s)}$, if $\text{supp}(Q_S)\supseteq\text{supp}(P_S)$, $+\infty$ otherwise.
   \item The conditional divergence is 
     \begin{equation*}
        D(P_{Y|S}\|Q_{Y|S}|P_S):= \E_{s\sim P_S}[D(P_{Y|S=s}\|Q_{Y|S=s})]. 
     \end{equation*}
   \item The mutual information of two random variables $S$ and $Y$ is $I(S;Y)=D(P_{SY}\|P_SP_Y)$. Equivalently, $I(S;Y)=D(P_{Y|S}\|P_Y|P_S)=D(P_{S|Y}\|P_S|P_Y)$. We use a subscript to specify the underlying distribution with respect to which the functionals are computed, e.g., $I_Q(S;Y|Z)$, under $Q=Q_{SYZ}$. 
   \item The conditional mutual information of $S$ and $Y$ given $Z$ is 
    \begin{align*}
      I_Q(S; Y | Z) = \sum_{z} Q_Z(z) \sum_{s,y}Q_{SY|Z}(s,y|z) \log\tfrac{Q_{SY|Z}(s,y|z)}{Q_{S|Z}(s|z)Q_{Y|Z}(y|z)}. 
    \end{align*}
\end{itemize}


\bibliography{R3paper}
\bibliographystyle{abbrvnat}

\end{document}